\documentclass[conference]{IEEEtran}
\IEEEoverridecommandlockouts
\usepackage{flushend}
\usepackage{amsmath,amssymb,amsthm}
\usepackage{dsfont}
\PassOptionsToPackage{hyphens}{url}
\usepackage[hidelinks]{hyperref}
\usepackage{breakurl}
\usepackage{url}
\usepackage{tikz}
\usepackage{subcaption}
\usepackage{caption}
\usepackage{graphicx}
\usepackage{multicol}
\usepackage{multirow}
\usepackage{amsthm}
\usepackage[compress]{cite}

\expandafter\def\expandafter\UrlBreaks\expandafter{\UrlBreaks
	\do\a\do\b\do\c\do\d\do\e\do\f\do\g\do\h\do\i\do\j%
	\do\k\do\l\do\m\do\n\do\o\do\p\do\q\do\r\do\s\do\t%
	\do\u\do\v\do\w\do\x\do\y\do\z\do\A\do\B\do\C\do\D%
	\do\E\do\F\do\G\do\H\do\I\do\J\do\K\do\L\do\M\do\N%
	\do\O\do\P\do\Q\do\R\do\S\do\T\do\U\do\V\do\W\do\X%
	\do\Y\do\Z\do\*\do\-\do\~\do\'\do\"\do\-}%

\newtheorem{theorem}{Theorem}
\newtheorem{remark}[theorem]{Remark}

\newtheorem{definition}[theorem]{Definition}
\newtheorem{corollary}[theorem]{Corollary}

\DeclareMathOperator{\atanh}{atanh}

\begin{document}

\author{Farhad Farokhi\thanks{F. Farokhi is with the CSIRO's Data61 and the University of Melbourne. e-mails: farhad.farokhi@data61.csiro.au; farhad.farokhi@unimelb.edu.au}\thanks{The work of F. Farokhi was funded, in part, by the Australian Government Department of the Environment and Energy through National Energy Analytics Research (NEAR) Program.}}	

\title{Temporally Discounted Differential Privacy for Evolving Datasets on an Infinite Horizon}

\maketitle

\begin{abstract}
	We define \textit{discounted differential privacy}, as an alternative to (conventional) differential privacy, to investigate privacy of \textit{evolving datasets, containing time series} over an unbounded horizon. We use \textit{privacy loss} as a measure of the amount of information leaked by the reports at a certain fixed time. We observe that privacy losses are weighted equally across time in the definition of differential privacy, and therefore the magnitude of privacy-preserving additive noise must grow without bound to ensure differential privacy over an infinite horizon. Motivated by the discounted utility theory within the economics literature, we use \textit{exponential and hyperbolic discounting of privacy losses} across time to relax the definition of differential privacy under continual observations. This implies that privacy losses in distant past are less important than the current ones to an individual. We use discounted differential privacy to investigate privacy of evolving datasets using additive Laplace  noise and show that the magnitude of the additive noise can remain bounded under discounted differential privacy. We illustrate the quality of privacy-preserving mechanisms satisfying discounted differential privacy on smart-meter measurement time-series of real households, made publicly available by Ausgrid (an Australian electricity distribution company).
\end{abstract}

\begin{IEEEkeywords}
	privacy; differential privacy; evolving datasets; temporal discounting.
\end{IEEEkeywords}

\section{Introduction}
Real-time analytics of customer data can benefit decision making of businesses in sectors, such as energy (e.g., real-time smart-meter measurements for demand and load forecasting), intelligent transportation (e.g., real-time traffic estimation by monitoring of movements of individuals), and retail industry (e.g., real-time analysis of customer interactions and purchases with online retail services to maximize profits). Privacy concerns, however, may restrict the availability of customer data or its use in decision making. For instance, smart-meter time-series can leak private information about household occupancy, entertainment habits, and air conditioning decisions~\cite{mcdaniel2009security, greveler2012multimedia,hosseini2017non}. The extent of privacy concerns have sometimes proved to  even hinder the roll out of smart meters~\cite{cuijpers2013smart}.

Differential privacy~\cite{dwork2014algorithmic,dwork2006calibrating, Dwork2006DP20972822097284} is a natural candidate to alleviate privacy concerns in general. However, differential privacy literature most often deals with providing privacy-preserving responses to queries based on large, yet static datasets that are kept securely by a data curator while, in real-time analytics, the underlying data in possession of the curator changes over time. The composition rule of differential privacy (see, e.g.,~\cite{dwork2014algorithmic}) implies that the magnitude of the additive noise that ensures differential privacy must grow rapidly, or that the privacy budget of each response must decrease correspondingly, to ensure that the entire privacy budget remains bounded. 

Recently, better performance bounds have been derived for differentially-private  responses to queries  on evolving datasets~\cite{joseph2018local,dwork2010differential, chan2011private,perrier2019private,cummings2018differential}. These studies however consider certain sets of queries, such as counting queries~\cite{dwork2010differential}. In these studies, the magnitude of the additive noise still remains unbounded; the best bound is $\mathcal{O}(\log(t)^{1.5})$ with $t$ denoting the number of observations, scaling with time. In fact, having a finite magnitude for the additive noise might not be even possible for some queries~\cite{vadhan2017complexity}.

Interesting approaches to differential privacy have stemmed from the study of dynamical systems in, e.g.,~\cite{le2012differentially, le2013differentially, cortes2016differential, wang2017differential, han2018privacy,le2017differentially}. However, they differ from the context of this paper in multiple directions. For instance, in~\cite{le2012differentially, le2013differentially}, the notion of neighbouring datasets requires that at most one component signal to differ and  this deviation must be bounded (in norm or count) while in this paper the deviation does not need to be bounded as the horizon over which we consider the problem can tend to infinity\footnote{Note that, here, we assume that each deviation is bounded but the norm of the signal can be unbounded as we do not consider finite horizon deviations.}. In~\cite{wang2017differential, han2016differentially}, finite horizon frameworks are considered. Although potentially infinite horizon algorithms were considered  in~\cite{nozari2015differentially,mo2016privacy, nozari2017differentially}, only privacy of the initial condition of the algorithm is investigated and thus the notion of neighbouring datasets only requires  single deviations. 

Motivated by these observations, in this paper, we relax the definition of the differential privacy by discounting privacy losses in distant past to be able to ensure privacy of evolving datasets over a possibly infinite\footnote{Note that, in practice, an infinite horizon does not exists. However, studying infinite horizons helps us to understand cases in which an upper bound on the horizon is not known. For instance, in the smart-metering example, we may not know, in advance, the duration for which a dataset of measurements is curated and made available for real-time analytic. } horizon. We use the notion of \textit{privacy loss} from~\cite{dwork2014algorithmic} to break down the definition of differential privacy across time. Privacy loss can be seen as a measure of the amount of information leaked by the reports at a certain fixed time. To ensure differential privacy, the summation of all privacy losses across time must be bounded by the privacy budget. We show that, because in the definition of differential privacy, the privacy losses are weighted equally across time, the magnitude of the privacy-preserving additive noise must grow unbounded to ensure differential privacy over an infinite horizon. Therefore, the reports become meaningless after a while. This motivates discounting privacy losses across time to generalize, or better-said relax, the definition of differential privacy. 

Discounting losses or gains across time is common place in the economics literature~\cite{myerson1995discounting,berns2007intertemporal}. A common practice is to exponentially discount losses or gains across time, i.e., scaling them by $\alpha^k$, where $\alpha$ is the discount factor and $k$ is the delay (time to or since the observation of the loss or gain). This discounting regime dates back to the early 20th century and is motivated by interest/cash rates~\cite{ramsey1928mathematical, 1023072967612}. More recently, it has been shown that humans and animals follow a hyperbolic discounting regime~\cite{myerson1995discounting, ainslie1975specious, ainslie1974impulse, Ainslie1981,kirby1997bidding}. In hyperbolic discounting, losses or gains are scaled by $1/(1+\beta k)$, where $\beta$ is the discounting coefficient and $k$ is again the delay. Hyperbolic discounting has been found to relate to real-world examples of self-control~\cite{vuchinich1998hyperbolic}, which makes it more interesting within the context of privacy preservation as it has been observed that privacy and self-control are heavily related in personal decision making~\cite{derlega1977privacy}.

The idea that privacy loss in a distant past is less important to an individual, and is thus discounted exponentially or hyperbolically, could be motivated by that people change their habits and addresses across time. For instance, according to the  \cite{abshouse} 2007-08 Survey of Income and Housing, 43\% of people in Australia have moved house within the last five years. Therefore, private information regarding their previous location can be deemed less sensitive. Temporal discounting is conjectured to be one of the reasons behind why few individuals take no action to protect their personal information, even when doing so involves limited costs~\cite{acquisti2004privacy}. Note that privacy behaviour (what people do) and privacy attitudes (what people think) should not be mistaken with each other as most individuals express that they are concerned about their  information privacy and are willing to act to protect it~\cite{acquisti2004privacy, rosen2011right}, yet in experiments they do not~\cite{Chellappa2005, spiekermann2001privacy}. Discounting privacy is more compatible with privacy behaviour in comparison to privacy attitudes. Finally, note that the discounting factor or coefficient can be chosen based on the preferences of an individual such that its effect is negligible over a long enough, yet finite, horizon, e.g., a person's life expectancy or active life span, albeit at the risk of reducing the quality of reports. This points to bigger issue of privacy and utility trade-off.

Discounting or decaying privacy was first studied in~\cite{bolot2013private}. However, the discounting in that paper was not motivated from a human's perception of loss in the economics theory. Therefore, in that paper, only exponential discounting was considered. As stated earlier, exponential discounting does not fully capture a human's response to loss across time. In this paper, in addition to exponential discounting, we consider hyperbolic discounting of privacy losses that is a better match for a human's perception of loss in distant past. An alternative to discounting or decaying privacy losses is to only consider a recent window, which was studied in~\cite{kellaris2014differentially,chen2017pegasus}. These approaches are of interest if we know the underlying population that generates the evolving dataset changes over time in regular intervals (e.g., WiFi users in a pubic domain over twenty-four hours window of time) and do not match the perception of a single individual from privacy loss over an entire horizon, which can be better matched by continuous discounting rather than windowing. 

In summary, in this paper, we make the following contributions:
\begin{itemize}
	\item We use exponential and hyperbolic discounting of privacy losses across time to relax the definition of differential privacy for use with evolving datasets;
	\item We use discounted differential privacy to investigate privacy of evolving datasets using an additive Laplace  noise and show that exponentially discounted differential privacy can be achieved with bounded magnitude of additive noise in contrast with differential privacy;
	\item We demonstrate the applicability of this paper's privacy-preserving mechanisms, which meet the standards of discounted differential privacy, on smart-meter measurement time series of private households, available from~\cite{ausgrid}.
\end{itemize}


The rest of the paper is organized as follows. In Section~\ref{sec:privacy_loss}, we define privacy loss and illustrate its relationship with differential privacy. We relax the notion of differential privacy to take into account temporal discounting of privacy losses in Section~\ref{sec:discounted}. We illustrate the results of the paper on smart-meter measurements of households in Section~\ref{sec:experiment}. We conclude the paper and present avenues for future research in Section~\ref{sec:conclusions}. 

\section{Differential Privacy and Privacy Loss} \label{sec:privacy_loss}
In this section, we introduce differential privacy for evolving datasets. We breakdown the definition of differential privacy across time by using privacy loss as a measure of the amount of information leaked by the reports at a certain fixed time. 
\subsection{Differential Privacy}
Consider an evolving, longitudinal dataset, containing time series, of the form 
\begin{align} \label{eqn:dataset}
X(t)\hspace{-.03in}:=\hspace{-.03in}
\begin{bmatrix}
x_1(1) & x_1(2) & \cdots & x_1(t)\\
x_2(1) & x_2(2) & \cdots & x_2(t)\\
\vdots & \vdots & \ddots & \vdots\\
x_n(1) & x_n(2) & \cdots & x_n(t)
\end{bmatrix} \hspace{-.03in}\in\hspace{-.03in}\mathcal{X}^{n\times t}\subseteq\mathbb{R}^{n\times t},
\end{align}
where $x_i(t)\in\mathcal{X}\subseteq\mathbb{R}$ denotes the entry for individual $i\in\{1,\dots,n\}$ at time instant $t\in\mathbb{N}$.  In this paper, we assume that $t$ is not bounded from above and can potentially approach infinity, i.e., it can grow unbounded.  Note that, for any $1\leq k\leq t$, $X(k)$ denotes a matrix extracted by eliminating  the last $t-k$ columns of $X(t)$. An example of such a longitudinal dataset is a dataset containing regular smart meter reading of $n$ fixed households. 

\begin{remark}[Addition/Removal of Individuals] In the dataset model in~\eqref{eqn:dataset}, we can consider addition and removal of individuals to the dataset across time. In this case, for all time instants in which a measurement for an individual is not available because it has not yet arrived or has left the dataset, we can use a special characters, such as  $\emptyset$. Hence, we must have $\mathcal{X}=\mathbb{R}\cup\{\emptyset\}$.
\end{remark}

We assume that, as the dataset evolves, the custodian reports 
\begin{align}
Y(t)
=
\begin{bmatrix}
y(1) & y(2) & \cdots & y(t)
\end{bmatrix}
\in\mathcal{Y}^{t}\subseteq\mathbb{R}^t.
\end{align}
Similarly, for any $1\leq k\leq t$, $Y(k)$ denotes a row vector extracted by eliminating the last $t-k$ entries of $Y(t)$. At time instance $t$, to generate the entry $y(t)$, the custodian uses conditional probability density function $p_{y|X}(\cdot|\cdot)$. From now on, we refer to this as the mechanism of the curator.  The mechanisms are causal by construct, that is, at any time instant $t$, the report $Y(t)$ can only be a function of the entries of the longitudinal dataset up to time $t$, $X(t)$.  In what follows, when it is evident from the context, we use $p(\cdot|\cdot)$ instead of $p_{y|X}(\cdot|\cdot)$. 
In this paper, we are interested in the differential privacy as a notion of privacy.

\begin{definition}[Neighbouring Datasets] Two datasets $X(t)$ and $X'(t)$ are neighbouring datasets, shown by $X(t)\sim X'(t)$, if they differ from each other in at most one row.
\end{definition}

\begin{definition}[Differential Privacy] A reporting mechanism is $\epsilon$-differentially private for $\epsilon>0$ if for any pair of neighbouring datasets $X(t),X'(t)$ and any output $Y(t)$,
\begin{align}
p(Y(t)|X(t))\leq \exp(\epsilon)p(Y(t)|X'(t)).
\end{align}
\end{definition}
Now, we are ready to introduce the notion of privacy loss from~\cite{dwork2014algorithmic} by de-constructing the ratio of probability density functions for use within the definition of differential privacy.

\subsection{From Differential Privacy to Privacy Loss}
In this subsection, we dig deeper in the notion of differential privacy to define privacy loss. 
Assuming conditional independence of $Y(k)$ and $Y(k-1)$ given $X(t)$ for all $2\leq k\leq t$, we get
\begin{align*}
p(Y(t)|X(t))
=&p(y(t)|X(t),Y(t-1)) p(Y(t-1)|X(t))\\
=&p(y(t)|X(t))p(Y(t-1)|X(t))\\
=& p(y(t)|X(t))p(y(t-1)|X(t))\\& \times p(Y(t-2)|X(t)),
\end{align*}
where the equalities follow from the definition of conditional probability density function. Following this line of reasoning, we get
\begin{align*}
p(Y(t)|X(t))
&=\prod_{k=1}^t p(y(k)|X(t))\\
&=\prod_{k=1}^t p(y(k)|X(k)),
\end{align*}
where the last equality follows from the causality of the reports. This results in 
\begin{align*}
\frac{p(Y(t)|X(t))}{p(Y(t)|X'(t))}
=\prod_{k=1}^t \frac{p(y(k)|X(k))}{p(y(k)|X'(k))}.
\end{align*}
These derivations motivate the use of $p(y(k)|X(k))/p(y(k)|X'(k))$, or in fact its logarithm, as a measure of privacy loss because if the ratio $p(y(k)|X(k))/p(y(k)|X'(k))$ is large, $y(k)$ leaks more information in terms of increasing the required differential-privacy budget for time instant $k$. 

\begin{definition}[Privacy Loss] Privacy loss due to entry $y(k)$ is 
\begin{align*}
\rho(k)=\sup_{y(k)}\sup_{X(k),X'(k):X(k)\sim X'(k)}\log\bigg(\frac{p(y(k)|X(k))}{p(y(k)|X'(k))}\bigg).
\end{align*}
\end{definition}

We can relate the notion of differential privacy and privacy loss together. This is explored in the next theorem.

\begin{theorem}[Privacy Loss and Differential Privacy] \label{tho:1} A reporting mechanism is $\epsilon$-differentially private for $\epsilon>0$ if
\begin{align}
\sum_{k=1}^t \rho(k)\leq \epsilon.
\end{align}
\end{theorem}

\begin{proof} We have
\begin{align*}
\frac{p(Y(t)|X(t))}{p(Y(t)|X'(t))}
&= \prod_{k=1}^t \frac{p(y(k)|X(k))}{p(y(k)|X'(k))}\\
&\leq \prod_{k=1}^t \exp(\rho(k))\\
&=\exp\left( \sum_{k=1}^t \rho(k)\right).
\end{align*}
The rest of the proof follows from the definition of differential privacy. 
\end{proof}

Theorem~\ref{tho:1} states that if the summation of all privacy losses is bounded from above by the total privacy budget $\epsilon$, $\epsilon$-differential privacy can be established. 

Now, consider the case where the curator is given a family of queries $f_t:\mathcal{X}^{n\times t}\rightarrow\mathbb{R}$, $\forall t\in\mathbb{N}$, to compute on the evolving dataset. In return, the curator provides noisy reports of the form:
\begin{align} \label{eqn:2}
y(t)=f_t(X(t))+w(t),
\end{align}
where $(w(t))_{t\in\mathbb{N}}$ is a sequence of i.i.d.\footnote{i.i.d. stands for independently and identically distributed. } Laplace random variables with zero mean and scale $b_t>0$.

\begin{theorem} \label{tho:2} The reporting mechanism~\eqref{eqn:2} is $\epsilon$-differentially private for $\epsilon>0$ if
\begin{align}
\sum_{k=1}^t\frac{\Delta f_k}{b_k}\leq \epsilon,
\end{align}
where $\Delta f_k$ is the sensitivity of the query defined as
\begin{align}
\Delta f_k:=\sup_{X(k),X'(k):X(k)\sim X'(k)} |f_k(X'(k))-f_k(X(k))|.
\end{align}
\end{theorem}

\begin{proof} Note that
\begin{align*}
\exp(\rho(k))
&=\sup_{y(k)}\sup_{X(k)\sim X'(k)}\frac{\exp(-|y(k)-f_k(X(k))|/b_k)}{\exp(-|y(k)-f_k(X'(k))|/b_k)}\\
&=\sup_{y(k)}\sup_{X(k)\sim X'(k)}\exp\left(\frac{1}{b_k}|y(k)-f_k(X'(k))|\right.\\
&\hspace{1.3in}\left. -\frac{1}{b_k}|y(k)-f_k(X(k))|\right)\\
&\leq\sup_{X(k)\sim X'(k)}\exp\left(\frac{1}{b_k}|f_k(X'(k))-f_k(X(k))|\right)\\
&\leq \exp\left(\frac{\Delta f_k}{b_k}\right),
\end{align*}
where the third equality follows from that $|a|-|b|\leq |a-b|$ because $|a|=|b-(b-a)|\leq |b|+|a-b|$. Hence,
$\sum_{k=1}^t \rho(k)
=\sum_{k=1}^t \Delta f_k/b_k.$
The rest follows from the application of Theorem~\ref{tho:1}.
\end{proof}

Theorem~\ref{tho:2} implies, with a reporting mechanism in the form of~\eqref{eqn:2}, we might not be able to ensure $\epsilon$-differential privacy over an unbounded horizon unless the magnitude of the additive noise grows unbounded or the queries become gradually less intrusive by decreasing $\Delta f_k$ rapidly enough.  This is because, if $b_t$ is kept constant and $\Delta f_k$ does not decrease, $\sum_{k=1}^t \Delta f_k/b_k=+\infty$, which makes the satisfaction of the condition of Theorem~\ref{tho:2} impossible. However, such a result might not be necessary as Theorem~\ref{tho:2} is only sufficient. Nonetheless this seems to be line with the differential privacy literature~\cite{vadhan2017complexity,dwork2010differential}.

\begin{corollary}
Assume that $\Delta f_k=\Delta f$ for all $k\in\mathbb{N}$. Then $\sum_{k=1}^t \Delta f_k/b_k\leq \epsilon$ for $\epsilon>0$ only if $\lim_{k\rightarrow}b_k=+\infty$. 
\end{corollary}

\begin{proof} Assume $\lim_{k\rightarrow}b_k=+\infty$ does not hold. Therefore, there exists a subsequence $(k_\ell)_{\ell\in\mathbb{N}}$ such that $k_\ell$ are increasing and $b_{k_\ell}\leq B$. Therefore, $\sum_{k=1}^\infty \Delta f_k/b_k\geq \sum_{\ell=1}^\infty \Delta f/b_{k_\ell}\geq \sum_{\ell=1}^\infty \Delta f/B=\infty$. Therefore, there exists a large enough $t_0$ for which $\sum_{k=1}^{t_0} \Delta f_k/b_k\leq \epsilon$ cannot be satisfied for any $t\geq t_0$. 
\end{proof}

\begin{corollary} \label{cor:1} Assume that $\Delta f_k=\Delta f$ for all $k\in\mathbb{N}$. The reporting mechanism~\eqref{eqn:2} is $\epsilon$-differentially private for $\epsilon>0$ if 
\begin{align}
b_k=\frac{\Delta f \pi^2k^2}{6\epsilon},\quad \forall k\in\mathbb{N}.
\end{align}
\end{corollary}

\begin{proof} Let $b_k=bk^2$. Hence, $\sum_{k=1}^t \Delta f_k/b_k\leq \sum_{k=1}^\infty \Delta f_k/b_k=(\Delta f/b)\pi^2/6$. Therefore, we can satisfy the condition of Theorem~\ref{tho:2} for any $t$ if $(\Delta f/b)\pi^2/6\leq \epsilon$. 
\end{proof}

These corollaries show that it might be necessary for the magnitude of the privacy-preserving additive noise to grow unbounded if we want to ensure differential privacy over an infinite horizon. This negative result could be caused by that the privacy loss at time instants $k$ and $t$ are weighted equally even if $k\ll t$ and, at $t$, the information leakage at time $k$ is no longer relevant. This motivates discounting privacy losses across time to relax the definition of differential privacy.

\section{Discounted Differential Privacy} \label{sec:discounted}
Humans perceive the severity (i.e., magnitude) of losses and gains differently across time. This fact is captured in the economics literature, especially, within expected utility theory, by discounting losses and gains that occur a long time from now.

\subsection{Exponentially Discounted Differential Privacy}
We start with exponentially discounted privacy loss and differential privacy. 

\begin{definition}[Exponentially Discounted Privacy Loss] At time instant $t$, privacy loss due to entry $y(k)$ is 
\begin{align*}
\varrho(k,t)
&=\alpha^{t-k}\rho(k),
\end{align*}
where $\alpha\in(0,1]$ is the discount factor. 
\end{definition}

Instead of privacy loss, we can use discounted privacy loss to ensure a certain level of privacy. Note that, at any time instant $t$, we have
$\sum_{k=1}^t \varrho(k,t)
=\sum_{k=1}^t\alpha^{t-k}\rho(k). $
Using this, we can define exponentially discounted differential privacy.

\begin{definition}[Exponentially Discounted Differential Privacy] A reporting mechanism is $(\epsilon,\alpha)$-exponentially discounted differentially private for $\epsilon>0$ and $\alpha\in(0,1]$ if 
\begin{align}
\sum_{k=1}^t\alpha^{t-k}\rho(k)\leq \epsilon.
\end{align}
\end{definition}

Note that $(\epsilon,1)$-exponentially discounted differential privacy is equivalent to $\epsilon$-differential privacy. 

\begin{theorem} \label{tho:3} The reporting mechanism~\eqref{eqn:2} is $(\epsilon,\alpha)$-exponentially discounted differentially private for $\epsilon>0$ and $\alpha\in(0,1]$ if
\begin{align}
\sum_{k=1}^t\alpha^{t-k}\bigg(\frac{\Delta f_k}{b_k}\bigg)\leq \epsilon.
\end{align}
\end{theorem}

\begin{proof} The proof follows the same line of reasoning as in Theorem~\ref{tho:2} while substituting the definition of  $\epsilon$-differential privacy with the definition of $(\epsilon,\alpha)$-exponentially discounted differential privacy. 
\end{proof}

\begin{corollary} \label{cor:2} Assume that $\Delta f_k=\Delta f$ for all $k\in\mathbb{N}$. The reporting mechanism~\eqref{eqn:2} is $(\epsilon,\alpha)$-exponentially discounted differentially private for $\epsilon>0$ and $\alpha\in(0,1)$ if 
\begin{align}
b_k=\frac{\Delta f }{\epsilon(1-\alpha)},\quad \forall k\in\mathbb{N}.
\end{align}
\end{corollary}

\begin{proof} Let $b_k=b$. Hence, $\sum_{k=1}^t \alpha^{t-k}\Delta f_k/b_k=(\Delta f/b)\sum_{k=0}^{t-1} \alpha^{k}\leq (\Delta f/b)\sum_{k=0}^{\infty} \alpha^{k}=(\Delta f/b)/(1-\alpha)$. Therefore, we can satisfy the condition of Theorem~\ref{tho:3} for any $t\in\mathbb{N}$ if $(\Delta f/b)/(1-\alpha)\leq \epsilon$. 
\end{proof}

Note that, in Corollary~\ref{cor:2}, the magnitude of the additive noise remains bounded. Therefore, the quality of the reports do not degrade with time, which is a drawback of adopting most notions of privacy for evolving datasets.

\subsection{Hyperbolic Discounted Differential Privacy}
As stated in the introduction, it has been shown that humans and animals follow a hyperbolic discounting regime. This motivate defining hyperbolic discounted differential privacy. 

\begin{definition}[Hyperbolic Discounted Privacy Loss] At time instant $t$, privacy loss due to entry $y(k)$ is 
\begin{align*}
\varrho'(k,t)
&=\frac{\rho(k)}{1+\beta(t-k)}
\end{align*}
where $\beta\geq 0$ is the discounting coefficient. 
\end{definition}

In this case, we have
$
\sum_{k=1}^t \varrho'(k,t)
=
\sum_{k=1}^t \rho(k)/(1+\beta k).
$
Using this, we can define hyperbolic discounted differential privacy.

\begin{definition}[hyperbolic Discounted Differential Privacy] A reporting mechanism is $(\epsilon,\beta)$-hyperbolic discounted differentially private for $\epsilon>0$ and $\beta>0$ if 
\begin{align}
\sum_{k=1}^t \frac{1}{1+\beta (t-k)} \rho(k)\leq \epsilon.
\end{align}
\end{definition}

Note that $(\epsilon,0)$-hyperbolic discounted differential privacy is equivalent to $\epsilon$-differential privacy. 

\begin{theorem} \label{tho:4} The reporting mechanism in~\eqref{eqn:2} is $(\epsilon,\beta)$-hyperbolic discounted differentially private for $\epsilon>0$ and $\beta\geq 0$ if
\begin{align}
\sum_{k=1}^t \frac{1}{1+\beta (t-k)}\frac{\Delta f_k}{b_k}\leq \epsilon.
\end{align}
\end{theorem}

\begin{figure*}[t]
	\centering
	\begin{tikzpicture}
	\node[] at (0,0) {\includegraphics[width=.75\columnwidth]{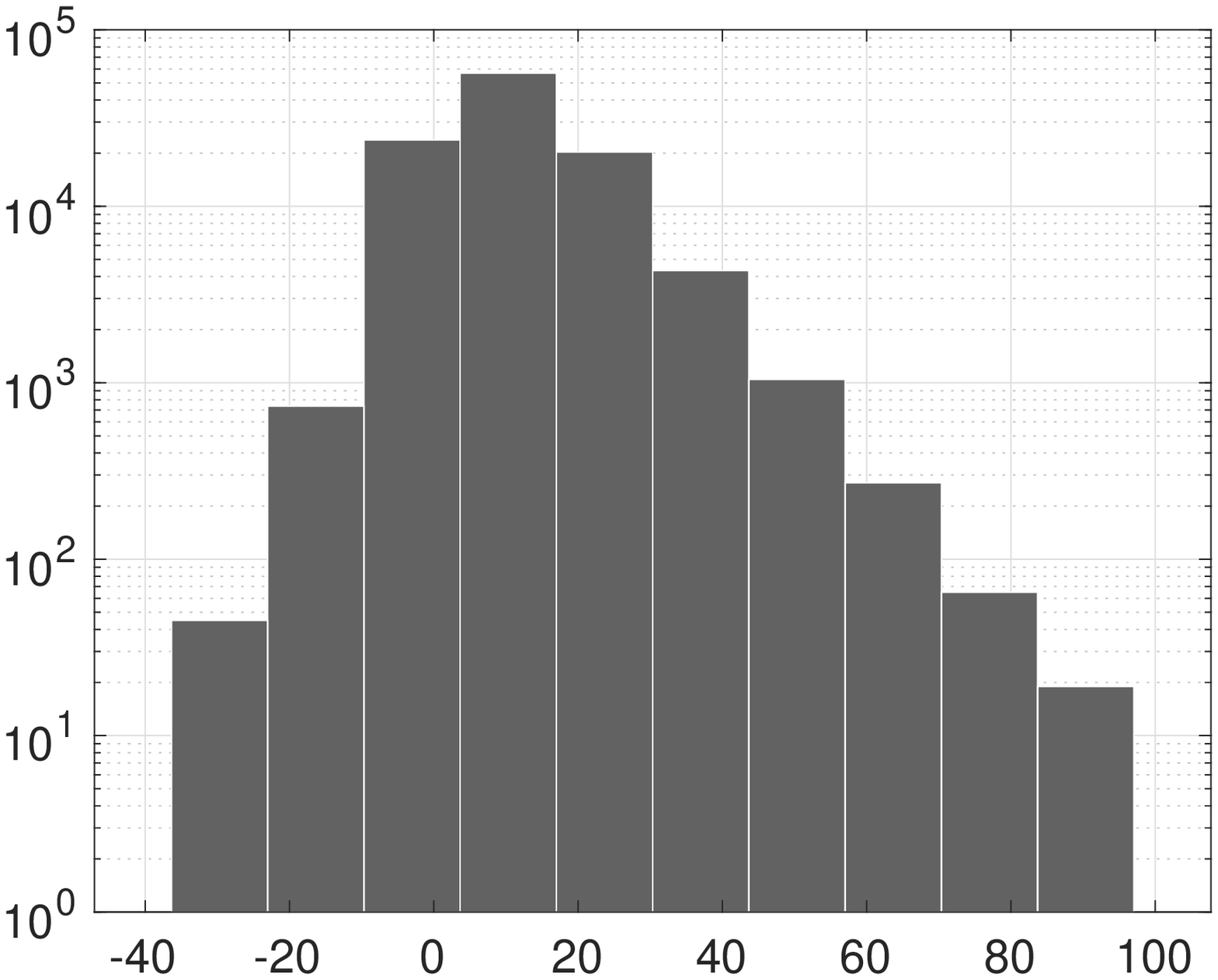}};
	\node[] at (0,-2.6) {\footnotesize smart-meter measurements (bins)};
	\node[rotate=90] at (-3.3,0) {\footnotesize frequency};
	\node[] at (7,0) {\includegraphics[width=.75\columnwidth]{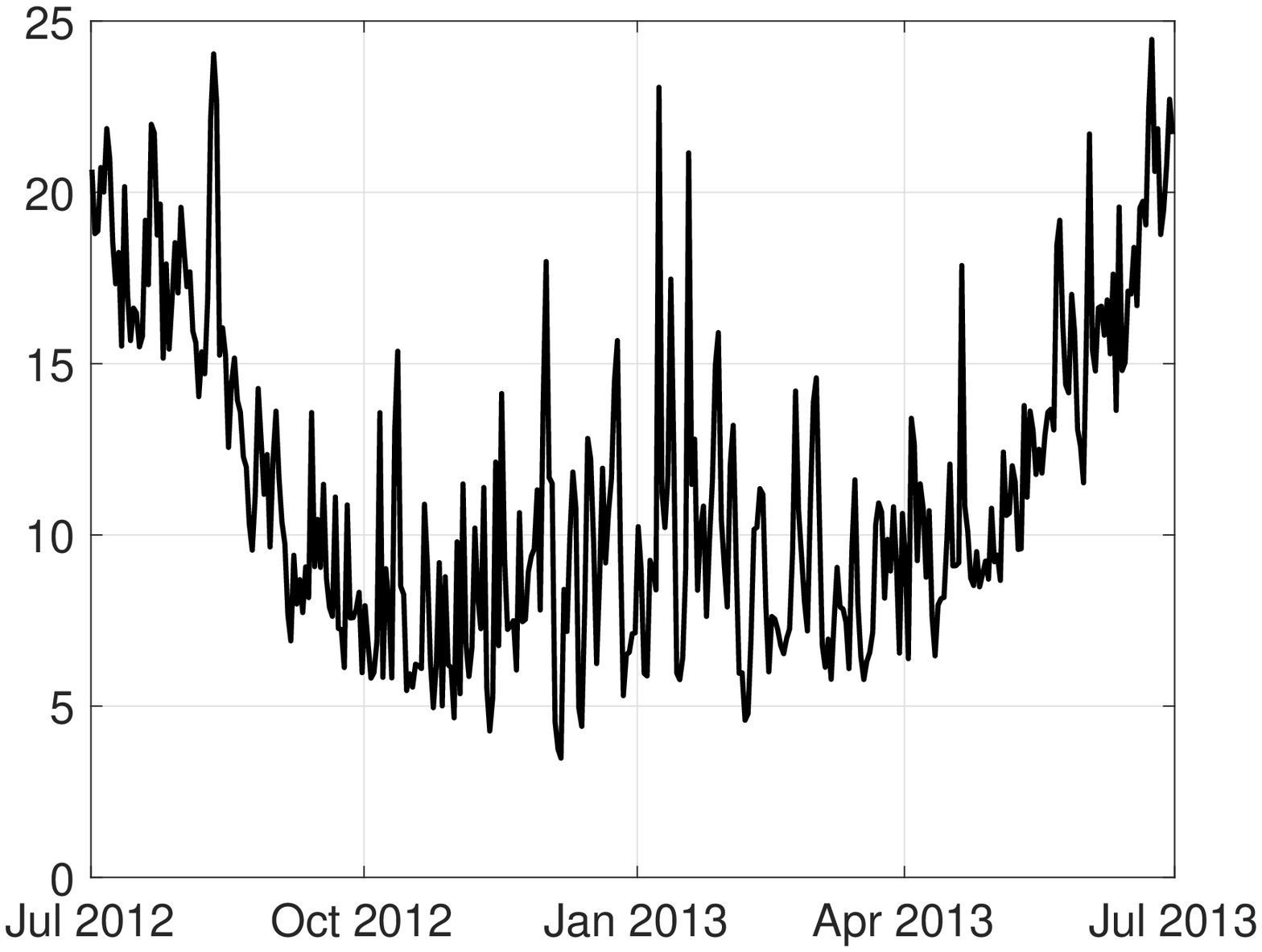}};
	\node[] at (7,-2.6)  {\footnotesize date};
	\node[rotate=90] at (+3.7,0) {\footnotesize daily average consumption};
	\end{tikzpicture}
	\caption{
		\label{fig:ausgrid_stat}
		Statistics of the Ausgrid data: [left] histogram of the smart-meter measurements of the households across the year and [right] average daily smart-meter measurements of the households. }
	\vspace{-.2in}
\end{figure*}

\begin{proof} The proof follows the same line of reasoning as in Theorem~\ref{tho:2} while substituting the definition of  $\epsilon$-differential privacy with the definition of $(\epsilon,\beta)$-hyperbolic discounted differential privacy. 
\end{proof}

\begin{corollary} \label{cor:3} Assume that $\Delta f_k=\Delta f$. The reporting mechanism in~\eqref{eqn:2} is $(\epsilon,\beta)$-hyperbolic discounted differentially private for $\epsilon>0$ and $\beta\geq 0$ if 
\begin{align}
b_k=\frac{\displaystyle
	2\Delta f\hspace{-.02in} \left(\hspace{-.03in}\atanh\left(\hspace{-.03in}\frac{1}{\sqrt{3}}\hspace{-.02in}\right)\hspace{-.03in}+\hspace{-.03in}\atanh\left(\hspace{-.03in}\sqrt{\frac{\beta}{1\hspace{-.02in}+\hspace{-.02in}\beta}}\right)\hspace{-.04in}\right)\hspace{-.03in}\sqrt{k}}{\epsilon\sqrt{\beta(\beta+1)}},\quad \hspace{-.03in}\forall k\in\mathbb{N}.
\end{align}
\end{corollary}

\begin{proof} Let $b_k=b\sqrt{k}$. By computing the derivatives, we can check that $\Delta f_k/((1+\beta (t-k))(b\sqrt{k}))$ is decreasing up to $(\beta t+1)/3\beta$ and is increasing afterwards. Let us define $t_0=\lfloor (\beta t+1)/(3\beta)\rfloor$. We have
\begin{align*}
\sum_{k=1}^{t_0}\frac{1}{1+\beta (t-k)}\frac{\Delta f_k}{b_k}
\leq& \frac{\Delta f}{b} \int_{0}^{t_0-1} \frac{1}{1+\beta (t-x)}\frac{1}{\sqrt{x}} \mathrm{d}x\\
=&\frac{2\Delta f}{b}\frac{\displaystyle\atanh\left(\sqrt{\frac{\beta (t_0-1)}{1+\beta t}}\right) }{\sqrt{\beta(\beta t +1)}}.
\end{align*}
Since $\atanh(\cdot)$ is an increasing function and $t_0-1\leq (\beta t+1)/(3\beta)$, we have
\begin{align*}
\atanh\left(\sqrt{\frac{\beta (t_0-1)}{1+\beta t}}\right)
\leq \atanh\left(\frac{1 }{\sqrt{3}}\right)
\end{align*}
We can also show that
\begin{align*}
\sum_{k=t_0+1}^{t}&\frac{1}{1+\beta (t-k)}\frac{\Delta f_k}{b_k}\\
\leq& \frac{\Delta f}{b} \int_{t_0+1}^t \frac{1}{(1+\beta (t-x))\sqrt{x}} \mathrm{d}x\\
=&\frac{2\Delta f}{b}\frac{\displaystyle\left(\atanh\left(\sqrt{\frac{\beta t}{1+\beta t}}\right) - \atanh\left(\sqrt{\frac{\beta(t_0+1)}{1+\beta t}}\right)\right)}{\sqrt{\beta(\beta t +1)}}\\
\leq &\frac{2\Delta f}{b}\frac{\displaystyle \atanh\left(\sqrt{\frac{\beta t}{1+\beta t}}\right) }{\sqrt{\beta(\beta t +1)}}.
\end{align*}
%
Note that
\begin{align*}
\frac{\mathrm{d}}{\mathrm{d}t}\frac{\displaystyle \atanh\left(\sqrt{\frac{\beta t}{1+\beta t}}\right) }{\sqrt{\beta(\beta t +1)}}
=&\frac{\sqrt{\beta t} }{\sqrt{t}(\beta t +1)^{3/2}}\left(
\sqrt{\frac{\beta t}{1+\beta t}}\right. \\
&\left.
-\atanh\left(\sqrt{\frac{\beta t}{1+\beta t}}\right) 
 \right)
 <0,
\end{align*}
where the inequality follows from that $\atanh(x)>x$ for all $x>0$. Therefore,
\begin{align*}
\frac{\displaystyle \atanh\left(\sqrt{\frac{\beta t}{1+\beta t}}\right) }{\sqrt{\beta(\beta t +1)}}\leq \frac{\displaystyle \atanh(\sqrt{\beta/(1+\beta)}) }{\sqrt{\beta(\beta+1)}}.
\end{align*}
This implies that 
\begin{align*}
\sum_{k=1}^{t}&\frac{1}{1+\beta (t-k)}\frac{\Delta f_k}{b_k}
\\&\leq \frac{2\Delta f}{b}\frac{\atanh(\sqrt{1/3})+\atanh(\sqrt{\beta/(1+\beta)})}{\sqrt{\beta(\beta+1)}}.
\end{align*}
The rest of the proof follows from the application of Theorem~\ref{tho:4}.
\end{proof}

\begin{figure*}[t]
	\centering
	\hspace{-.15in}
	\begin{tikzpicture}
	\node[] at (0,0) {\includegraphics[width=.7\columnwidth]{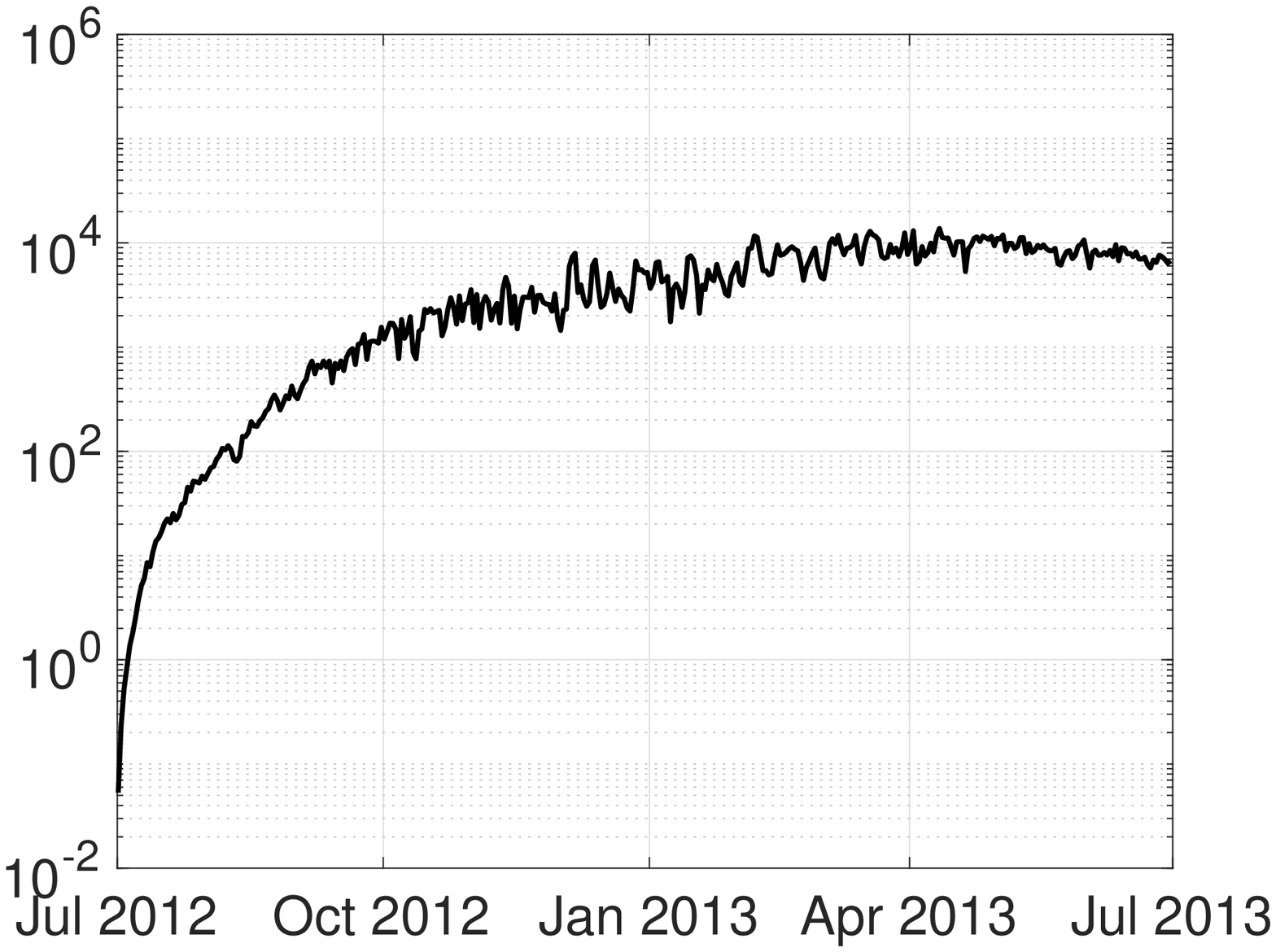}};
	\node[] at (0,-2.4) {\footnotesize date};
	\node[rotate=90] at (-2.95,0) {\footnotesize Expected relative error};
	\node[] at (6.1,0) {\includegraphics[width=.7\columnwidth]{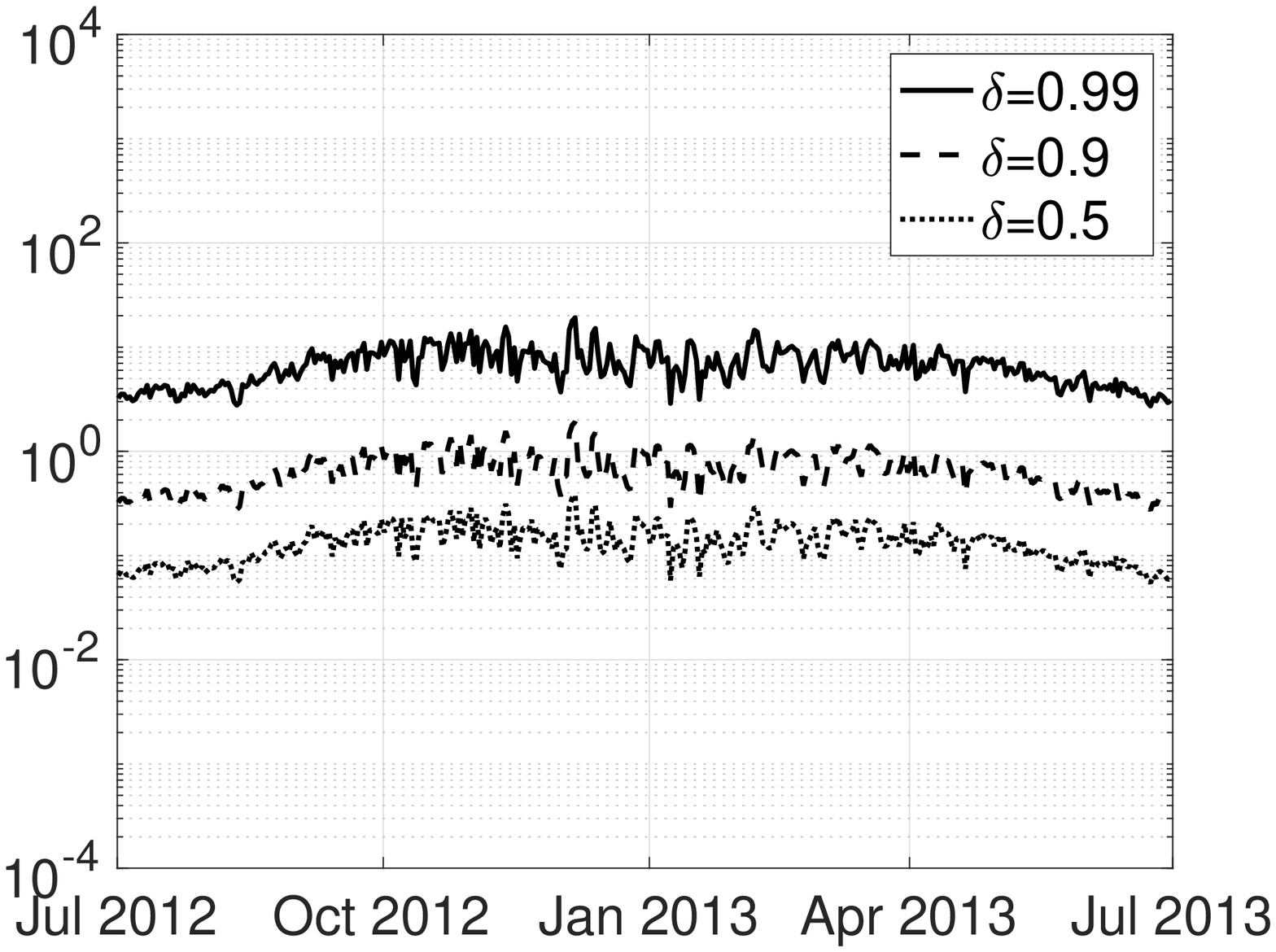}};
	\node[] at (6.1,-2.4) {\footnotesize date};
	\node[rotate=90] at (3.1,0) {\footnotesize Expected relative error};
	\node[] at (12.2,0) {\includegraphics[width=.7\columnwidth]{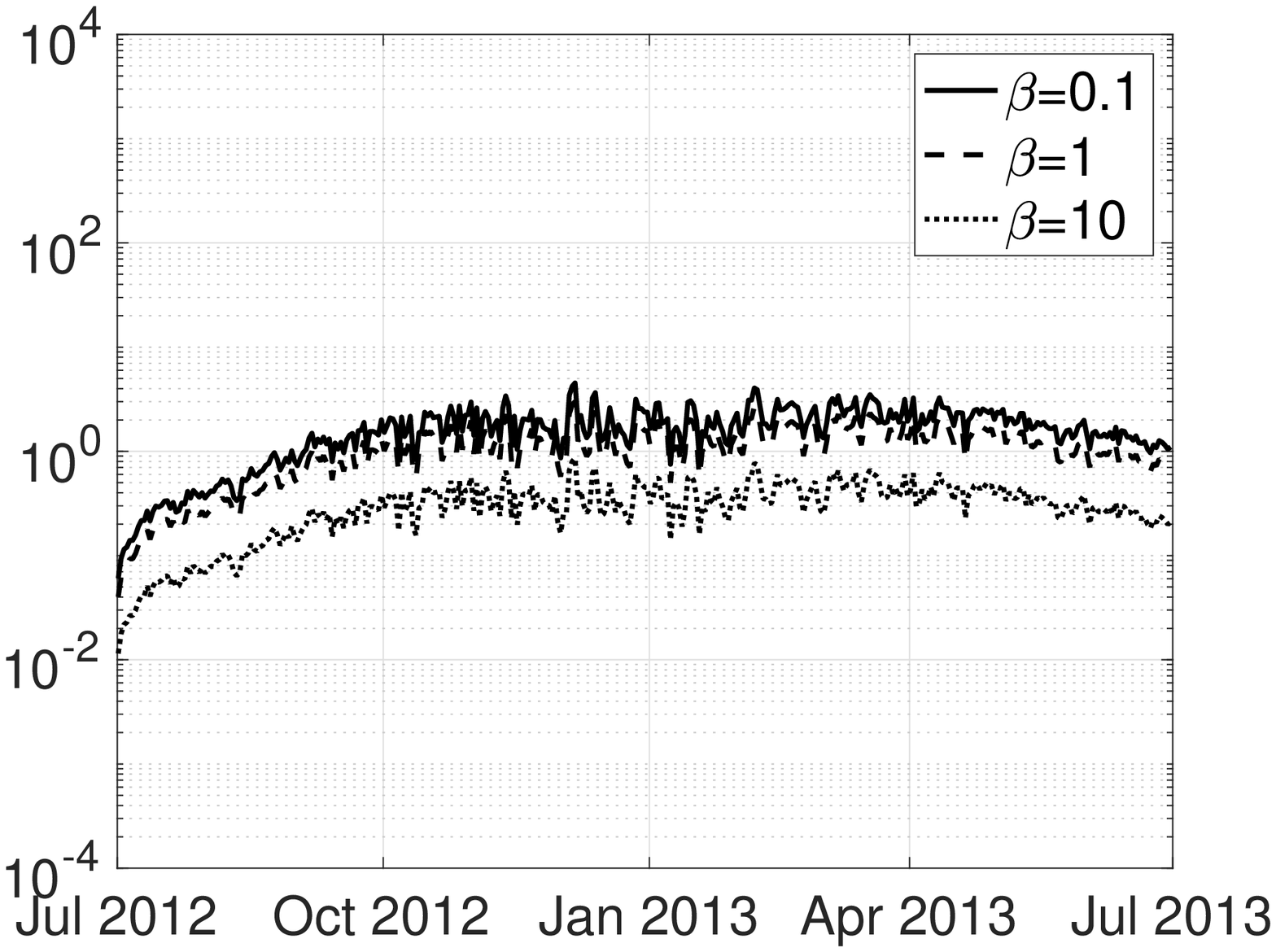}};
	\node[] at (12.2,-2.4) {\footnotesize date};
	\node[rotate=90] at (9.15,0) {\footnotesize Expected relative error};
	\end{tikzpicture}
	\caption{\label{fig:ausgrid_privacy} Quality of reports for the Ausgrid data under various notions of privacy: differential privacy [left], exponentially discounted differential privacy [middle], and hyperbolic discounted differential privacy [right].}
		\vspace{-.2in}
\end{figure*}
In Corollary~\ref{cor:3}, the magnitude of the additive noise grows unbounded but at a much slower rate than Corollary~\ref{cor:1}. Therefore, the quality of the reports, although degrading with time, remain better than differential privacy.

\begin{figure*}[t]
	\centering
	\begin{tikzpicture}
	\node[] at (0,0) {\includegraphics[width=.75\columnwidth]{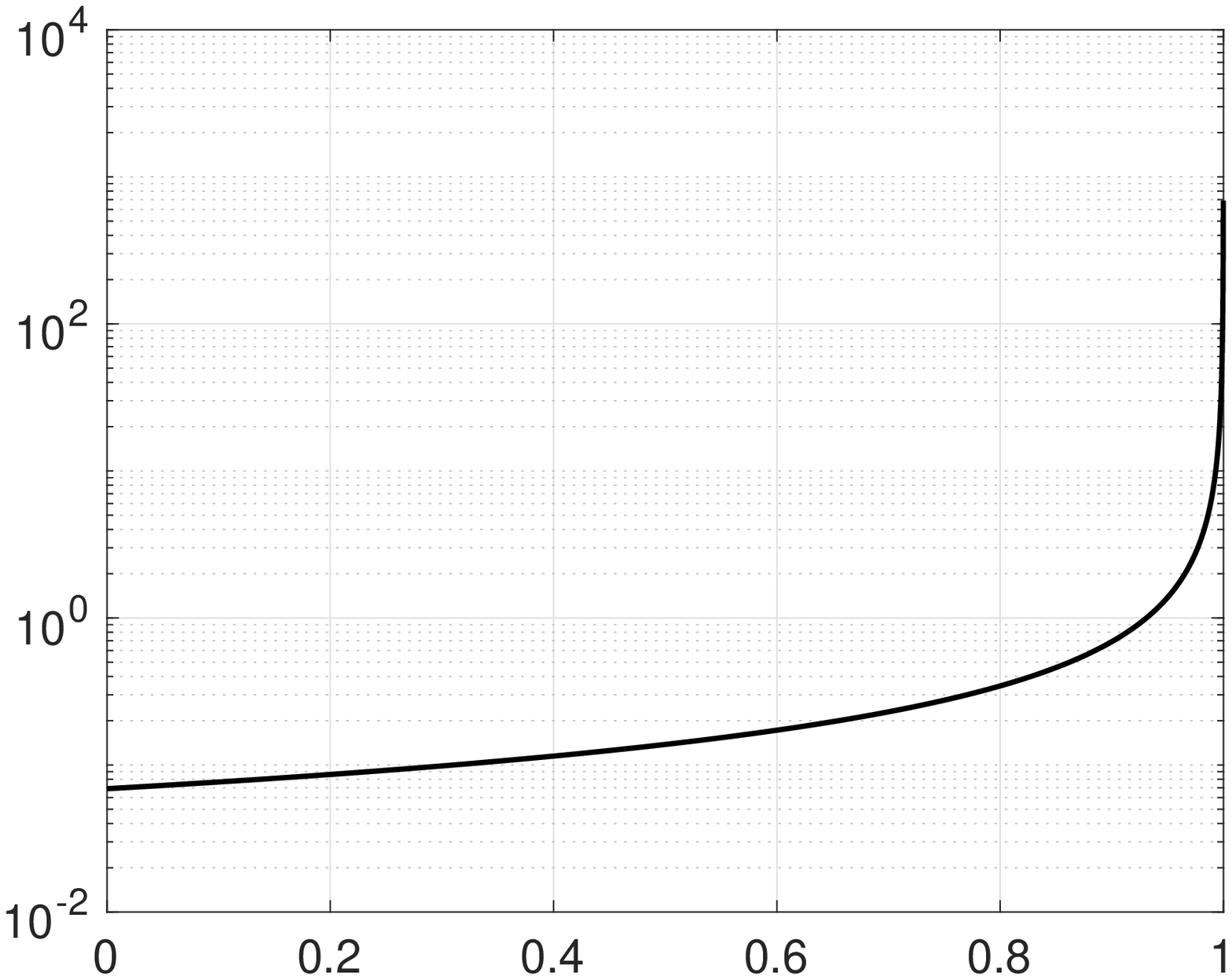}};
	\node[] at (0,-2.5) {\footnotesize discount factor $\alpha$};
	\node[rotate=90] at (-3.3,0) {\footnotesize Average expected relative error};
	\node[] at (7,0) {\includegraphics[width=.75\columnwidth]{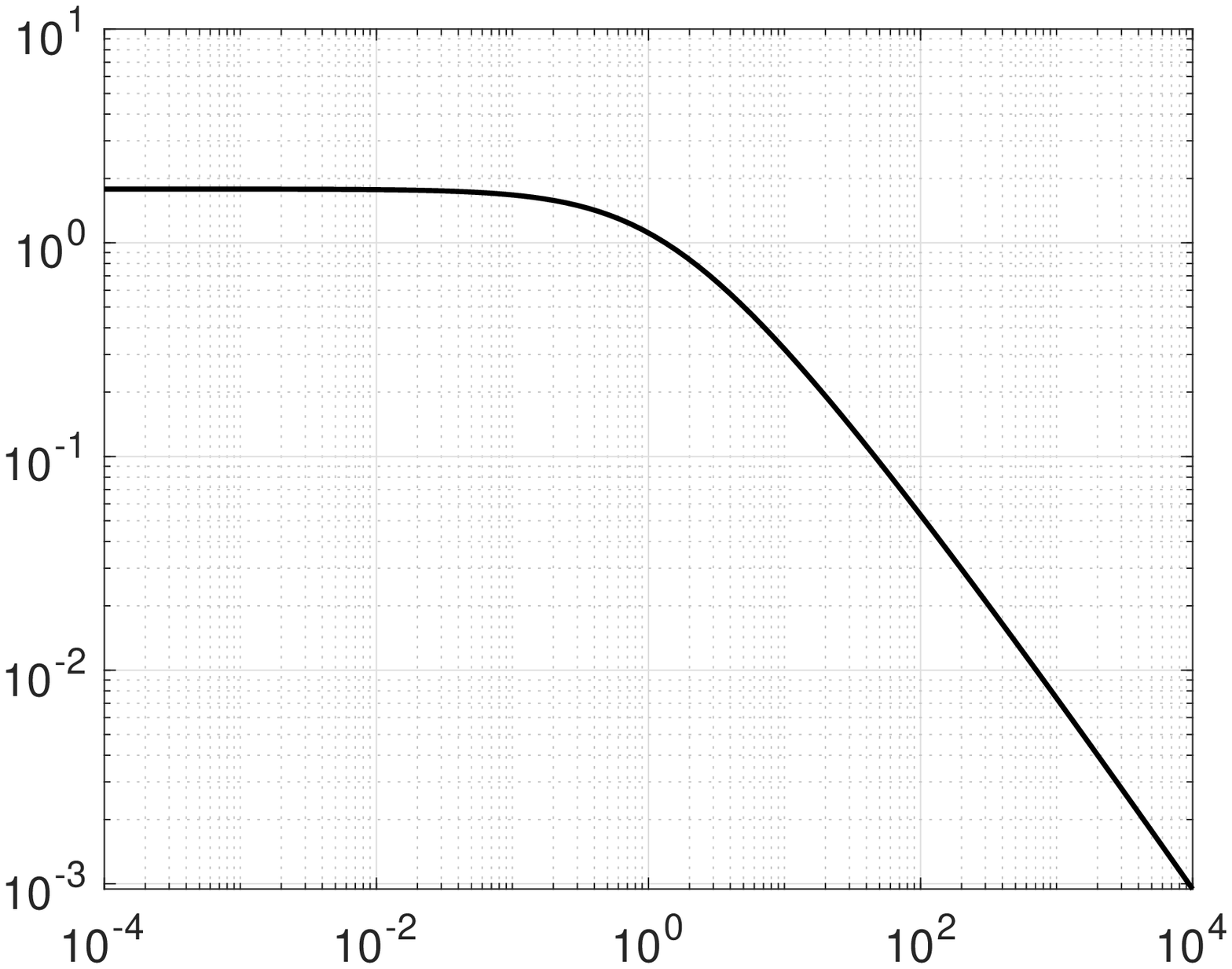}};
	\node[] at (7,-2.5) {\footnotesize discounting coefficient $\beta$};
	\node[rotate=90] at (3.5,0) {\footnotesize Average expected relative error};
	\end{tikzpicture}
	\caption{\label{fig:ausgrid_privacy_tradeoff}
		Quality of reports  for the Ausgrid data  versus [left] discount factor $\alpha$ and [right] discounting coefficient $\beta$. }
\end{figure*}

\section{Numerical Results} \label{sec:experiment}
In this section, we illustrate the results of the paper on regular smart-meter measurements from real households. We numerically investigate the quality of privacy-preserving reports ensuring discounted differential privacy using the expected difference between noisy privacy-preserving daily averages and true potentially privacy-intrusive average consumption of the households. 

\subsection{Data}
We use smart meter measurements of of households, made available by the~\cite{ausgrid} to illustrate the results of this paper. The individuals in both these datasets have been de-identified. The dataset in~\cite{ausgrid} contains electricity data for 300 homes with rooftop solar systems that are measured by a smart meter that, in addition to measuring the usage from the grid,  records the total amount of solar power generated. The measurements are obtained every 30 minutes over 2010-2013. In this paper, we use the data over July 2012 to June 2013. Figure~\ref{fig:ausgrid_stat} illustrates the statistics of the Ausgrid data. Figure~\ref{fig:ausgrid_stat}~[left] shows the histogram of the smart-meter measurements of the households across the year.  Figure~\ref{fig:ausgrid_stat}~[right] illustrates the  daily smart-meter measurements of the households, averaged across the individuals. This is the intended report that must to be released using the privacy-preserving mechanism in~\eqref{eqn:2}. Therefore, \vspace{-.05in}
\begin{align*}
f_t(X(t))=\frac{1}{n}\sum_{i=1}^n x_i(t).\vspace{-.05in}
\end{align*}
 It can be seen that, for this example, $\Delta f=200/300$ in which $200$ kWh is maximum changes in consumption across a household, according to Figure~\ref{fig:ausgrid_stat}~[left], and $300$ is the number of households. 

\subsection{Setup}
We present daily consumption of the households across the year, averaged for the individuals. We use the reporting mechanism~\eqref{eqn:2} to ensure the privacy of the households. We consider three setups of differential privacy, exponentially discounted differential privacy, and hyperbolic discounted differential privacy.  To this aim, we use the results of Corollaries~\ref{cor:1},~\ref{cor:2}, and~\ref{cor:3} to ensure that the reports are privacy preserving in their corresponding notions. We are interested in investigating the quality of the privacy-preserving reports, i.e., the difference between the reports and the average consumption of the households. Particularly, we use the expected relative error, defined as
\begin{align*}
\mbox{expected relative error}\;(t):=\frac{\displaystyle \mathbb{E}\left\{\left|y(t)-\frac{1}{n}\sum_{i=1}^n x_i(t)\right|\right\}}{\displaystyle \left|\frac{1}{n}\sum_{i=1}^n x_i(t)\right|},
\end{align*}
as a measure of quality. We also use average expected relative error, which is defined as
\begin{align*}
&\mbox{avrage expected relative error}\\
&\hspace{.8in}=\frac{1}{T}\sum_{t=1}^T \mbox{expected relative error}\;(t),
\end{align*}
where $T$ denotes the horizon of the experiment, one year in this paper.

\subsection{Results}
Figure~\ref{fig:ausgrid_privacy} shows the quality of reports in~\eqref{eqn:2} under various notions of privacy: differential privacy [left], exponentially discounted differential privacy [middle], and hyperbolic discounted differential privacy [right]. As expected, quality of reports for $\epsilon$-differential private reporting policy is bad and degrades with time, as the magnitude of the noise needs to be increased to ensure differential privacy. It is interesting to note that, in the hyperbolic discounted differential privacy, because the increase in the magnitude of the privacy-preserving noise is so low that we cannot observe its disruptive effect within 1 year. 

To quantify the effects of the discount factor and discounting coefficient in discounted differential privacy, we can study the average expected relative error in Figure~\ref{fig:ausgrid_privacy_tradeoff}. As expected, by increasing the discounting across time, which can be achieved by decreasing discount factor or increasing  discounting coefficient, the performance improves. However, the privacy guarantee also weakens as privacy losses from the past are dismissed faster, which might not be desirable. 

\section{Conclusions and Future Work} \label{sec:conclusions}
We defined discounted differential privacy to investigate privacy in the context of evolving datasets. We used exponential and hyperbolic discounting of privacy losses across time to relax the definition of differential privacy under continual observations. We used discounted differential privacy to investigate privacy of evolving datasets using an additive Laplace noise. We illustrate the quality of privacy-preserving mechanisms satisfying discounted differential privacy on smart-meter measurement. Future work can focus on capturing the effect of temporal discounting on the ability of an adversary to observe private information of an individual household. 

\section*{Acknowledgements}
The author is thankful to the anonymous reviewers for improving the presentation of the paper and for spotting some flaws in the preliminary proofs that was fixed in the final version of the paper.

\bibliographystyle{ieeetr}
\bibliography{citation}

\begin{thebibliography}{10}

\bibitem{mcdaniel2009security}
P.~McDaniel and S.~McLaughlin, ``Security and privacy challenges in the smart
  grid,'' {\em IEEE Security \& Privacy}, vol.~7, no.~3, pp.~75--77, 2009.

\bibitem{greveler2012multimedia}
U.~Greveler, B.~Justus, and D.~Loehr, ``Multimedia content identification
  through smart meter power usage profiles,'' {\em Computers, Privacy and Data
  Protection}, vol.~1, no.~10, 2012.

\bibitem{hosseini2017non}
S.~S. Hosseini, K.~Agbossou, S.~Kelouwani, and A.~Cardenas, ``Non-intrusive
  load monitoring through home energy management systems: A comprehensive
  review,'' {\em Renewable and Sustainable Energy Reviews}, vol.~79,
  pp.~1266--1274, 2017.

\bibitem{cuijpers2013smart}
C.~Cuijpers and B.-J. Koops, ``Smart metering and privacy in {Europe}: lessons
  from the {Dutch} case,'' in {\em European data protection: Coming of age},
  pp.~269--293, Springer, 2013.

\bibitem{dwork2014algorithmic}
C.~Dwork and A.~Roth, ``The algorithmic foundations of differential privacy,''
  {\em Foundations and Trends in Theoretical Computer Science}, vol.~9,
  no.~3--4, pp.~211--407, 2014.

\bibitem{dwork2006calibrating}
C.~Dwork, F.~McSherry, K.~Nissim, and A.~Smith, ``Calibrating noise to
  sensitivity in private data analysis,'' in {\em Theory of Cryptography
  Conference}, pp.~265--284, Springer, 2006.

\bibitem{Dwork2006DP20972822097284}
C.~Dwork, ``Differential privacy,'' in {\em Proceedings of the 33rd
  International Conference on Automata, Languages and Programming - Volume Part
  II}, ICALP'06, (Berlin, Heidelberg), pp.~1--12, Springer-Verlag, 2006.

\bibitem{joseph2018local}
M.~Joseph, A.~Roth, J.~Ullman, and B.~Waggoner, ``Local differential privacy
  for evolving data,'' in {\em Advances in Neural Information Processing
  Systems}, pp.~2375--2384, 2018.

\bibitem{dwork2010differential}
C.~Dwork, M.~Naor, T.~Pitassi, and G.~N. Rothblum, ``Differential privacy under
  continual observation,'' in {\em Proceedings of the 42nd ACM Symposium on
  Theory of Computing}, pp.~715--724, ACM, 2010.

\bibitem{chan2011private}
T.-H.~H. Chan, E.~Shi, and D.~Song, ``Private and continual release of
  statistics,'' {\em ACM Transactions on Information and System Security
  (TISSEC)}, vol.~14, no.~3, p.~26, 2011.

\bibitem{perrier2019private}
V.~Perrier, H.~J. Asghar, and D.~Kaafar, ``Private continual release of
  real-valued data streams,'' in {\em Proceedings of the Network and
  Distributed System Security Symposium (NDSS)}, 2019.

\bibitem{cummings2018differential}
R.~Cummings, S.~Krehbiel, K.~A. Lai, and U.~Tantipongpipat, ``Differential
  privacy for growing databases,'' in {\em Advances in Neural Information
  Processing Systems}, pp.~8864--8873, 2018.

\bibitem{vadhan2017complexity}
S.~Vadhan, ``The complexity of differential privacy,'' in {\em Tutorials on the
  Foundations of Cryptography}, pp.~347--450, Springer, 2017.

\bibitem{le2012differentially}
J.~Le~Ny and G.~J. Pappas, ``Differentially private kalman filtering,'' in {\em
  2012 50th Annual Allerton Conference on Communication, Control, and Computing
  (Allerton)}, pp.~1618--1625, IEEE, 2012.

\bibitem{le2013differentially}
J.~Le~Ny and G.~J. Pappas, ``Differentially private filtering,'' {\em IEEE
  Transactions on Automatic Control}, vol.~59, no.~2, pp.~341--354, 2013.

\bibitem{cortes2016differential}
J.~Cort{\'e}s, G.~E. Dullerud, S.~Han, J.~Le~Ny, S.~Mitra, and G.~J. Pappas,
  ``Differential privacy in control and network systems,'' in {\em 2016 IEEE
  55th Conference on Decision and Control (CDC)}, pp.~4252--4272, IEEE, 2016.

\bibitem{wang2017differential}
Y.~Wang, Z.~Huang, S.~Mitra, and G.~E. Dullerud, ``Differential privacy in
  linear distributed control systems: Entropy minimizing mechanisms and
  performance tradeoffs,'' {\em IEEE Transactions on Control of Network
  Systems}, vol.~4, no.~1, pp.~118--130, 2017.

\bibitem{han2018privacy}
S.~Han and G.~J. Pappas, ``Privacy in control and dynamical systems,'' {\em
  Annual Review of Control, Robotics, and Autonomous Systems}, vol.~1,
  pp.~309--332, 2018.

\bibitem{le2017differentially}
J.~Le~Ny and M.~Mohammady, ``Differentially private {MIMO} filtering for event
  streams,'' {\em IEEE Transactions on Automatic Control}, vol.~63, no.~1,
  pp.~145--157, 2017.

\bibitem{han2016differentially}
S.~Han, U.~Topcu, and G.~J. Pappas, ``Differentially private distributed
  constrained optimization,'' {\em IEEE Transactions on Automatic Control},
  vol.~62, no.~1, pp.~50--64, 2016.

\bibitem{nozari2015differentially}
E.~Nozari, P.~Tallapragada, and J.~Cort{\'e}s, ``Differentially private average
  consensus with optimal noise selection,'' {\em IFAC-PapersOnLine}, vol.~48,
  no.~22, pp.~203--208, 2015.

\bibitem{mo2016privacy}
Y.~Mo and R.~M. Murray, ``Privacy preserving average consensus,'' {\em IEEE
  Transactions on Automatic Control}, vol.~62, no.~2, pp.~753--765, 2016.

\bibitem{nozari2017differentially}
E.~Nozari, P.~Tallapragada, and J.~Cort{\'e}s, ``Differentially private average
  consensus: Obstructions, trade-offs, and optimal algorithm design,'' {\em
  Automatica}, vol.~81, pp.~221--231, 2017.

\bibitem{myerson1995discounting}
J.~Myerson and L.~Green, ``Discounting of delayed rewards: Models of individual
  choice,'' {\em Journal of the Experimental Analysis of Behavior}, vol.~64,
  no.~3, pp.~263--276, 1995.

\bibitem{berns2007intertemporal}
G.~S. Berns, D.~Laibson, and G.~Loewenstein, ``Intertemporal choice--toward an
  integrative framework,'' {\em Trends in Cognitive Sciences}, vol.~11, no.~11,
  pp.~482--488, 2007.

\bibitem{ramsey1928mathematical}
F.~P. Ramsey, ``A mathematical theory of saving,'' {\em The Economic Journal},
  vol.~38, no.~152, pp.~543--559, 1928.

\bibitem{1023072967612}
P.~A. Samuelson, ``A note on measurement of utility,'' {\em The Review of
  Economic Studies}, vol.~4, pp.~155--161, 02 1937.

\bibitem{ainslie1975specious}
G.~Ainslie, ``Specious reward: A behavioral theory of impulsiveness and impulse
  control,'' {\em Psychological Bulletin}, vol.~82, no.~4, p.~463, 1975.

\bibitem{ainslie1974impulse}
G.~W. Ainslie, ``Impulse control in pigeons,'' {\em Journal of the Experimental
  Analysis of Behavior}, vol.~21, no.~3, pp.~485--489, 1974.

\bibitem{Ainslie1981}
G.~Ainslie and R.~J. Herrnstein, ``Preference reversal and delayed
  reinforcement,'' {\em Animal Learning \& Behavior}, vol.~9, no.~4,
  pp.~476--482, 1981.

\bibitem{kirby1997bidding}
K.~N. Kirby, ``Bidding on the future: Evidence against normative discounting of
  delayed rewards,'' {\em Journal of Experimental Psychology: General},
  vol.~126, no.~1, p.~54, 1997.

\bibitem{vuchinich1998hyperbolic}
R.~E. Vuchinich and C.~A. Simpson, ``Hyperbolic temporal discounting in social
  drinkers and problem drinkers,'' {\em Experimental and Clinical
  Psychopharmacology}, vol.~6, no.~3, p.~292, 1998.

\bibitem{derlega1977privacy}
V.~J. Derlega and A.~L. Chaikin, ``Privacy and self-disclosure in social
  relationships,'' {\em Journal of Social Issues}, vol.~33, no.~3,
  pp.~102--115, 1977.

\bibitem{abshouse}
{Australian Bureau of Statistics (ABS)}, ``Australian social trends {December}
  2010: Moving house,'' 2010.
\newblock
  \url{https://www.ausstats.abs.gov.au/ausstats/subscriber.nsf/LookupAttach/4102.0Publication14.12.104/$File/41020_housingmobility2010.pdf}.

\bibitem{acquisti2004privacy}
A.~Acquisti and J.~Grossklags, ``Privacy attitudes and privacy behavior,'' in
  {\em Economics of Information Security}, pp.~165--178, Springer, 2004.

\bibitem{rosen2011right}
J.~Rosen, ``The right to be forgotten,'' {\em Stanford Law Review}, vol.~64,
  p.~88, 2011.

\bibitem{Chellappa2005}
R.~K. Chellappa and R.~G. Sin, ``Personalization versus privacy: An empirical
  examination of the online consumer's dilemma,'' {\em Information Technology
  and Management}, vol.~6, no.~2, pp.~181--202, 2005.

\bibitem{spiekermann2001privacy}
S.~Spiekermann, J.~Grossklags, and B.~Berendt, ``E-privacy in 2nd generation
  e-commerce: Privacy preferences versus actual behavior,'' in {\em Proceedings
  of the 3rd ACM conference on Electronic Commerce}, pp.~38--47, ACM, 2001.

\bibitem{bolot2013private}
J.~Bolot, N.~Fawaz, S.~Muthukrishnan, A.~Nikolov, and N.~Taft, ``Private
  decayed predicate sums on streams,'' in {\em Proceedings of the 16th
  International Conference on Database Theory}, pp.~284--295, ACM, 2013.

\bibitem{kellaris2014differentially}
G.~Kellaris, S.~Papadopoulos, X.~Xiao, and D.~Papadias, ``Differentially
  private event sequences over infinite streams,'' {\em Proceedings of the VLDB
  Endowment}, vol.~7, no.~12, pp.~1155--1166, 2014.

\bibitem{chen2017pegasus}
Y.~Chen, A.~Machanavajjhala, M.~Hay, and G.~Miklau, ``Pegasus: Data-adaptive
  differentially private stream processing,'' in {\em Proceedings of the 2017
  ACM SIGSAC Conference on Computer and Communications Security},
  pp.~1375--1388, ACM, 2017.

\bibitem{ausgrid}
Ausgrid, ``Solar home electricity data,'' 2014.
\newblock
  \url{https://www.ausgrid.com.au/Industry/Innovation-and-research/Data-to-share/Solar-home-electricity-data}.

\end{thebibliography}

\end{document}